\documentclass[11pt]{article}

\newif\ifnotes
\notestrue

\usepackage{hyperref}
\usepackage{fullpage}
\usepackage{amssymb}
\usepackage{amsmath}
\usepackage{amsthm}
\usepackage{amsfonts}
\usepackage{bm}
\usepackage{enumitem}
\usepackage{color}
\usepackage{comment}
\usepackage[capitalize]{cleveref}
\usepackage[dvipsnames]{xcolor}
\usepackage{float}
\usepackage[T1]{fontenc}
\usepackage{todonotes}
\usepackage{asymptote}
\usepackage{mdframed}
\usepackage[most]{tcolorbox}
\usepackage{hyperref}
\usepackage{enumitem}
\usepackage{framed}
\usepackage{mdframed}
\usepackage{scrextend}
\usepackage{bbm}
\usepackage{thm-restate}

\usepackage[T1]{fontenc}

\definecolor{denim}{rgb}{0.08, 0.38, 0.74}
\definecolor{periwinkle}{rgb}{0.6, 0.6, 0.95}
\definecolor{wildblueyonder}{rgb}{0.64, 0.68, 0.82}

\usepackage{hyperref}
\hypersetup{
    colorlinks=true,
    linkcolor=denim,
    filecolor=denim,
    citecolor=periwinkle,
    urlcolor=periwinkle,
}
\usepackage[hyperpageref]{backref}

\newtheorem{theorem}{Theorem}[section]
\newtheorem{definition}[theorem]{Definition}
\newtheorem{lemma}[theorem]{Lemma}
\newtheorem{claim}[theorem]{Claim}

\newtheorem{corollary}[theorem]{Corollary}

\theoremstyle{remark}

\Crefname{theorem}{Theorem}{Theorems}
\Crefname{claim}{Claim}{Claims}
\Crefname{lemma}{Lemma}{Lemmas}
\Crefname{proposition}{Proposition}{Propositions}
\Crefname{corollary}{Corollary}{Corollaries}
\Crefname{definition}{Definition}{Definitions}

\newcommand{\ECC}{\mathsf{ECC}}

\newcommand{\iECC}{\mathsf{iECC}}

\newcommand{\maj}{\text{maj}}

\newcommand{\Diam}{\text{Diam}}

\newcommand{\bbN}{\mathbb{N}}
\newcommand{\bbR}{\mathbb{R}}

\makeatletter
\newcommand{\customlabel}[2]{%
   \protected@write \@auxout {}{\string \newlabel {#1}{{#2}{\thepage}{#2}{#1}{}} }%
   \hypertarget{#1}{#2}
}
\makeatother

\newcommand{\attack}[2]{
    \stepcounter{figure}
    \vspace{0.15cm}
    { \small
    \begin{tcolorbox}[breakable, enhanced, colback=wildblueyonder!20]
    \begin{center}
    {\bf \underline{Attack~\customlabel{attack:#1}{\thefigure}}}
    \end{center}
    
    #2
    \end{tcolorbox}
    }
}

\newcounter{casenum}

\newcounter{subcasenum}

\newcounter{casenump}

\newcommand{\casep}[2]{
    \ifthenelse{\equal{\value{casenump}}{0}}{
    \vskip.5\baselineskip\par\noindent
    }{}
    {\it Case \arabic{casenump}:} {\it #1}
    \vskip0.1\baselineskip
    \begin{addmargin}[1.5em]{1em}
    #2
    \end{addmargin}
    \addtocounter{casenump}{1}
}

\newcounter{subcasenump}

\begin{document}

\title{A New Upper Bound on the Maximal Error Resilience of Interactive Error-Correcting Codes}
\author{Meghal Gupta \thanks{E-mail:\texttt{meghal@berkeley.edu}. This author was supported by an UC Berkeley Chancellor’s fellowship.}\\UC Berkeley \and Rachel Yun Zhang\thanks{E-mail:\texttt{rachelyz@mit.edu}. This research was supported in part by DARPA under Agreement No. HR00112020023, an NSF grant CNS-2154149, and NSF Graduate Research Fellowship 2141064.}\\MIT}
\date{\today}

\sloppy
\maketitle
\begin{abstract}
In an \emph{interactive error-correcting code ($\iECC$)}, Alice and Bob engage in an interactive protocol with the goal of Alice communicating a message $x \in \{ 0, 1 \}^k$ to Bob in such a way that even if some fraction of the total communicated bits are corrupted, Bob can still determine $x$. It was shown in works by Gupta, Kalai, and Zhang (STOC 2022) and by Efremenko, Kol, Saxena, and Zhang (FOCS 2022) that there exist $\iECC$'s that are resilient to a larger fraction of errors than is possible in standard error-correcting codes without interaction. 

One major question in the study of $\iECC$'s is to determine the optimal error resilience achievable by an $\iECC$. In the case of bit flip errors, it is known that an $\iECC$ can achieve $\frac14 + 10^{-5}$ error resilience (Efremenko, Kol, Saxena, and Zhang), while the best known upper bound is $\frac27 \approx 0.2857$ (Gupta, Kalai, and Zhang). In this work, we improve upon the upper bound, showing that no $\iECC$ can be resilient to more than $\frac{13}{47} \approx 0.2766$ fraction of errors.
\end{abstract}
\thispagestyle{empty}
\newpage
\tableofcontents
\pagenumbering{roman}
\newpage
\pagenumbering{arabic}

\section{Introduction}

Consider the following task: Alice wishes to communicate a message to Bob such that even if a constant fraction of the communicated bits are adversarially tampered with, Bob is still guaranteed to be able to determine her message. This task motivated the prolific study of \emph{error-correcting codes}, starting with the seminal works of~\cite{Shannon48, Hamming50}. An error-correcting code encodes a message $x$ into a longer codeword $\ECC(x)$, such that the Hamming distance between any two distinct codewords is a constant fraction of the length of the codewords. To communicate a message $x$, Alice sends Bob the corresponding codeword $\ECC(x)$, and the fact that the distance between any two codewords is large guarantees that an adversary must corrupt a large fraction of the communication in order for Bob to decode to the wrong codeword. 

An important question in the study of error-correcting codes is determining the maximal possible error resilience, that is, the maximal possible number of bits such that as long as an adversary does not flip more than that number of bits, Bob is guaranteed to decode to the right $x$. Indeed, many works focus on precisely this question.  In general, the error resilience parameter depends on the alphabet size: in this work, we focus on the binary alphabet. It is well known that in the adversarial bit-flip model, no $\ECC$ can be resilient to more than $\frac14$ corruptions.

\paragraph{Interactive Error-Correcting Codes.}

Recently, the work \cite{GuptaKZ22} proposed using \emph{interaction} to improve the error resilience past $\frac14$. They define the model of an \emph{interactive error-correcting code ($\iECC$)} as follows. In an $\iECC$, Alice and Bob engage in a fixed length, fixed order protocol where Alice's goal is to communicate a message $x$ to Bob. The adversary is given a corruption budget which is an $\alpha$-fraction of the total communication (for some $\alpha>0$). She can spend it arbitrarily (e.g., all on forward communication, or on some combination of forward and feedback communication). The goal is for Bob to learn $x$ no matter how the adversary corrupts the communicated bits.

A priori, it is not clear that an $\iECC$ should allow one to improve the error resilience past $\frac14$. Indeed, the presence of Bob's bits increase the adversary's corruption allowance relative to Alice's forward communication, while not having clear benefit (as the adversary could choose to corrupt all of Bob's feedback to an adversarial string, thereby potentially causing more harm than good). Nevertheless, the work of~\cite{GuptaKZ22} gave the first evidence that interaction has the power to buff error resilience: they demonstrate an $\iECC$ in the case of adversarial \emph{erasures} that is resilient to more erasures than standard (non-interactive) $\ECC$'s can possible be resilient to.

In the case of bit flip errors, the following work of~\cite{EfremenkoKSZ22} demonstrated the first evidence that interaction is useful against bit flip errors as well: their $\iECC$ achieves resilience to $\frac14 + 10^{-5}$ errors, more than the maximal possible error resilience of $\frac14$ achievable by any $\ECC$. This prompts the question: \emph{What is the optimal error resilience achievable by an $\iECC$?} 

The current best known upper (impossibility) bound for this problem was given in~\cite{GuptaKZ22}, who showed that no $\iECC$ can be resilient to more than $\frac27$ adversarial errors. This upper bound came from the combination of two natural attacks, one of which is guaranteed to work no matter how the rounds in which Alice and Bob speak are distributed.
\begin{enumerate}[label=\textbf{Attack \arabic*:}, leftmargin=*]
    \item 
        Corrupt \emph{none} of Bob's bits. Then, Bob's messages provide perfect reliable feedback, in which case works about \emph{error-correcting codes with feedback} beginning with~\cite{Berlekamp64} tell us that it suffices to corrupt $\frac13$ of Alice's bits. 
    \item 
        Corrupt \emph{half} of Bob's bits so that his messages appear random and thus essentially are useless: then Alice's communication essentially reduces to the case of a standard error-correcting code, in which case an adversary can corrupt $\frac14$ of her bits to confuse Bob between two possible values of $x$.
\end{enumerate}

Nevertheless, the question remained: \emph{What is the largest possible error resilience of an $\iECC$? Is it possible to achieve error resilience equal to this natural upper bound of $\frac27$?}

In this work, we answer the latter question in the negative, providing a new upper bound of $\frac{13}{47} \approx 0.2766$, improving upon the previous best upper bound of $\frac27 \approx 0.2857$. 

\begin{theorem}[Main Result]
    For sufficiently small $\epsilon>0$, there exists $k_0 = k_0(\epsilon) \in \bbN$ such that for any $k > k_0$, no $\iECC$ over the binary bit flip channel where Alice is trying to communicate $x \in \{ 0, 1 \}^k$ is resilient to $\frac{13}{47}+\epsilon$ fraction of adversarial bit flips.
\end{theorem}

\section{Related Works}

In this section, we discuss previous work on interactive error-correcting codes, as well as prior work on error-correcting codes with feedback.

\subsection{Interactive Error-Correcting Codes}

The notion of an \emph{interactive error-correcting code} ($\iECC$) was first introduced in \cite{GuptaKZ22}, who demonstrated an $\iECC$ resilient to $\frac35$ fraction of adversarial erasures, surpassing the best possible erasure resilience of standard $\ECC$'s of $\frac12$. They also gave an upper bound of $\frac23$ on the erasure resilience of any $\iECC$. In the case of bit flip errors, they proved an upper bound of $\frac27$ on the error resilience achievable by any $\iECC$, leaving open the problem of constructing an $\iECC$ resilient to greater than $\frac14$ adversarial errors. 

The followup work of~\cite{GuptaZ22B} improved upon the erasure $\iECC$ of~\cite{GuptaKZ22}, giving a construction of an $\iECC$ with \emph{positive rate} but resilient to only $\frac{6}{11}$ adversarial erasures.
% The model is recent, but the intuition for the model comes from the field of interactive coding (see e.g., \cite{Gelles-survey}) originated by \cite{Schulman92}. 
% The authors propose the problem of constructing an $\iECC$ in the erasure model resilient to $\alpha>\frac12$ and an $\iECC$ in the bit flip errors model resilient to $\alpha>\frac14$ erasures. 

% In the erasures model, \cite{GuptaKZ22} shows an $\iECC$ resilient to $\frac35$ erasures and proves that no $\iECC$ can be resilient to more than $\frac23$ erasures. Later, \cite{GuptaZ22B} gives a positive rate construction, albeit only with error resilience $\frac6{11}$. The optimal error resilience remains unknown.

In the bit flip error model, \cite{EfremenkoKSZ22} answered~\cite{GuptaKZ22}'s question in the affirmative, constructing an $\iECC$ with error resilience $\frac14+10^{-5}$. This narrowed the optimal error resilience of any $\iECC$ to the range $[\frac14 + 10^{-5}, \frac27]$. In this paper, we further narrow this range, improving the upper bound from $\frac27$ to $\frac{13}{47}$.

\subsection{Error-Correcting Codes with Feedback}\label{sec:related:feedback}

The use of interaction in the noise resilient communication of a message has been studied previously in the form of \emph{error-correcting codes with feedback}. In an error-correcting code with feedback, Alice wishes to communicate a message to Bob in an error-resilient fashion, provided that after every message she sends she receives some feedback from Bob about what he has received. She can then use this noiseless feedback to choose the next bit that she sends. Error-correcting codes with feedback were first introduced in the Ph.D. thesis of Berlekamp~\cite{Berlekamp64} and have been studied in a number of followup works, including~\cite{Berlekamp68,Zigangirov76,SpencerW92,HaeuplerKV15,GuptaGZ22,AhlswedeDL06}. 
Originally, this feedback was considered in the \emph{noiseless} setting, meaning that none of Bob's messages are allowed to be corrupted, and error rate is calculated solely as a function of the number of messages Alice sends. That is, Bob's feedback is free and always correct, so that Alice can tailor her next message to specifically the bit of information Bob most needs to hear.

% Error-correcting codes with feedback over a binary alphabet are known to be equivalent to a noisy version of the game \emph{Twenty Questions}, or the Bar Kochba game~\cite{Renyi61}. Variations of the Bar Kochba game 

In the bit flip error model, \cite{Berlekamp68,Zigangirov76,SpencerW92,HaeuplerKV15} showed that the maximal error resilience of an error-correcting code with noiseless feedback is $\frac13$. \cite{GuptaGZ22} show this is achievable even by protocols that only send logarithmically many bits of feedback over a constant number of rounds. For explicit constant number of rounds of feedback, \cite{BravermanEKSZ22} initiated the study of the noise resilience vs. round complexity tradeoff for both erasures and errors. For larger alphabets, the maximal error resilience was studied in~\cite{AhlswedeDL06}.

When the feedback is \emph{noisy}, i.e. the feedback may be corrupted as well, much less is known. Several works such as \cite{BurnashevY08a, BurnashevY08b} considered $\ECC$'s with noisy feedback over the binary symmetric channel. 
% (each communicated bit is independently flipped with some probability). 
\cite{WangQC17} considers adversarial corruption, under a model which places separate corruption budgets on the forward and feedback rounds. They construct a scheme that is resilient to $\frac12$ of the forward communication and $1$ of the feedback being erased. We note that their scheme's forward erasure resilience is equal to that achievable by standard error-correcting codes. 
\section{Preliminaries}

\paragraph{Notation.} In this paper, we use the following notation:
\begin{itemize}
    \item For a string $x$, $x[i]$ denotes the $i$'th bit of $x$, and $x[i:j]$ denotes the substring of $x$ starting with the $i$'th bit and ending with the $j$'th bit.
    \item For any $n \in \bbN$, $[n]$ denotes the integers $1, \dots, n$. For any $n, m \in \bbN$, $[n, m]$ denotes the integers $n, \dots, m$. 
    \item The \emph{diameter} of three strings $s_1, s_2, s_3$, denoted $\Diam(s_1, s_2, s_3)$, is the maximal Hamming distance between any two of the three strings.
    \item The majority function, denoted $\maj$, takes as input some number of bits and outputs the most frequent bit. (If there are an equal number of $0$'s and $1$'s, it outputs either.)
\end{itemize}

\subsection{Interactive Error-Correcting Codes} 

We formally define the notion of an interactive error-correcting code ($\iECC$). 

\begin{definition} [Interactive Error-Correcting Code]
    An interactive error-correcting code ($\iECC$) is a non-adaptive interactive protocol\footnote{A nonadaptive interactive protocol is a fixed-length, fixed-order of speaking protocol between two parties where in each round, a predetermined party sends a single bit to the other.} $\pi = \{ \pi_k \}_{k \in \bbN}$, with the following syntax: 
    \begin{itemize}
        \item At the beginning of the protocol, Alice receives as private input some $x \in \{ 0, 1 \}^k$.
        \item At the end of the protocol, Bob outputs some $\hat{x} \in \{ 0, 1 \}^k$.
    \end{itemize}
    We say that $\pi$ is $\alpha$-error resilient if there exists $k_0 \in \bbN$ such that for all $k > k_0$ and $x \in \{ 0, 1 \}^k$, and for all online adversarial attacks consisting of flipping at most $\alpha \cdot |\pi|$ of the total communication, Bob outputs $x$ at the end of the protocol with probability $1$.
\end{definition}

\subsection{Important Lemmas} 

We now state some important lemmas and combinatorial theorems. 

\begin{lemma} \label{lem:clique-dist-12}
    Among any $K$ strings $s_1, \dots, s_K \in \{ 0, 1 \}^\ell$, there exist two strings with Hamming distance at most $\left( \frac12 + \frac{1}{2(K-1)} \right) \cdot \ell$. 
\end{lemma}

\begin{proof}
    For any index $\iota \in [\ell]$, let $c_\iota$ be the number of strings $s_j$ for which $s_j[\iota] = 0$. Then, 
    \begin{align*}
        \sum_{(i, j) \in \binom{K}{2}} \Delta(s_i, s_j) 
        = \sum_{\iota \in [\ell]} c_\iota(K - c_\iota) 
        \le \ell \cdot \frac{K^2}{4},
    \end{align*}
    so by the pigeonhole principle, there exists $(i, j) \in \binom{K}{2}$ such that $\Delta(s_i, s_j) \le \frac{\ell \cdot K^2/4}{K(K-1)/2} = \left( \frac12 + \frac{1}{2(K-1)} \right) \cdot \ell$.
\end{proof}

\begin{theorem}\cite{Shearer83} \label{thm:shearer}
    There exists $\delta_0$ such that for any $\delta < \delta_0$, in any undirected graph on $n$ vertices with at most $\delta n^3$ triangles, there is an independent set of size $\frac{1}{4\delta}$. 
\end{theorem}

\begin{corollary} \label{cor:shearer}
    For small enough $\epsilon > 0$ and for any collection of $K$ strings $s_1, \dots, s_K \in \{ 0, 1 \}^\ell$, there are $\epsilon K^3/4$ unordered triples $(s_i, s_j, s_k)$ of distinct strings such that $\Diam(s_i, s_j, s_k) \le \left( \frac12 + \epsilon \right) \cdot \ell$. 
\end{corollary}

\begin{proof}
    Consider the graph with the $K$ strings as vertices, such that there is an edge between $s_i$ and $s_j$ if and only if $\Delta(s_i, s_j) \le \left( \frac12 + \epsilon \right) \cdot \ell$. If there are fewer than $\epsilon K^3/4$ triangles, then by Theorem~\ref{thm:shearer} there is an independent set of size $1/\epsilon$, which by Lemma~\ref{lem:clique-dist-12} is not possible.
\end{proof}

\begin{theorem}(Turan, \cite{Turan41}) \label{thm:turan}
    For any $m \in \bbN$, any graph on any vertices with less than $\frac{n^2}{2m}$ edges has an independent set of size $m$.
\end{theorem}

\begin{corollary} \label{cor:turan}
    For small enough $\epsilon > 0$ and for any collection of $K$ strings $s_1, \dots, s_K \in \{ 0, 1 \}^{\ell}$, there are at least $\epsilon K^2/2$ unordered pairs of strings $(s_i, s_j)$ such that $\Delta(s_i, s_j) \le \left( \frac12 + \epsilon \right) \cdot \ell$.
\end{corollary}

\begin{proof}
    Consider the graph with the $K$ strings as vertices, such that there is an edge between $s_i$ and $s_j$ if $\Delta(s_i, s_j) \le \left( \frac12 + \epsilon \right) \cdot \ell$. If there are fewer than $\epsilon K^2/2$ unordered pairs of strings $(s_i, s_j)$ for which $\Delta(s_i, s_j) \le \left( \frac12 + \epsilon \right) \cdot \ell$, then by Theorem~\ref{thm:turan} there is an independent set of size $\frac1\epsilon$, which by Lemma~\ref{lem:clique-dist-12} is not possible.
\end{proof}

\begin{theorem}(Ramsey, \cite{Ramsey87}) \label{thm:ramsey}
    For any $m,n \in \bbN$, there exists $R(m,n) \le \binom{m+n-2}{n-1} \in \bbN$ such that any undirected graph with at least $R(m, n)$ vertices has either a clique on $m$ verticles or an independent set on $n$ vertices.
\end{theorem}

\section{Impossibility Bound on Maximal Noise Resilience of $\iECC$} \label{sec:13/47}

In this section, we will present our main result, that for any non-adaptive $\iECC$, there is some attack consisting of at most $\frac{13}{47}$ corruptions such that Bob cannot guess Alice's input $x$ correctly with probability better than $\frac12$. 

\begin{restatable}{theorem}{main} \label{thm:13/47}
    For sufficiently small $\epsilon > 0$, then for all $k>100\epsilon^{-4}$, no $\iECC$ for $x\in \{0,1\}^k$ is resilient to more than $\frac{13}{47}+2\epsilon$ fraction of errors with probability greater than $\frac12$.
\end{restatable}

The rest of this section will be devoted to the proof of this theorem. Throughout this section, Alice's input will always be denoted $x \in \{ 0, 1 \}^k$. The length of the $\iECC$ will be denoted by $n$. 

At a high level, our proof will proceed as follows. We will split any candidate protocol into two sections, the first consisting of the first $\frac{21}{47} n$ rounds of the protocol, and the second consisting of the remaining $\frac{26}{47} n$ rounds of the protocol. In the first section, we denote the number of bits that Alice sends by $A_1$, and the number that Bob sends by $B_1$. Likewise, in the second section, we denote the number of bits that Alice and Bob send by $A_2$ and $B_2$ respectively. We will present three attacks in Sections~\ref{sec:attack1},~\ref{sec:attack2}, and~\ref{sec:attack3} such that depending on the values of $A_1, B_1, A_2, B_2$, at least one attack is guaranteed to succeed while using at most $\frac{13}{47}$ corruptions. 

Throughout this section, a \emph{transcript} is the sequence of bits that is received by either of the parties. Note that since the adversary may corrupt messages, the transcript may be different than what was sent by Alice and Bob. We say that an attack \emph{succeeds} with $\alpha$ corruption if there exist two inputs $x_1, x_2 \in \{ 0, 1 \}^k$ along with respective strategies corrupting at most $\alpha n$ bits such that Bob's view of the transcript in both cases is identical. Then, Bob cannot guess Alice's true value of $x \in \{ x_1, x_2 \}$ with probability better than $\frac12$.

% In this section, we show our main result that no $\iECC$ is resilient to $>\frac{13}{47}$ errors. We divide any candidate protocol into two sections: the first $\frac{21}{47}$ bits of the protocol and the last $\frac{26}{47}$ bits of the protocol. In the first section, Alice sends $A_1$ bits and Bob sends $B_1$ bits. In the second section, Alice sends $A_2$ bits and Bob sends $B_2$ bits. 

\subsection{Attack 1} \label{sec:attack1}

In the first attack, the adversary behaves the same on both sections of the protocol. She corrupts Alice's bits while leaving Bob's untouched, such that there exist two inputs for which at most of $\frac13$ of Alice's communication is corrupted. We remark that this attack has been known since~\cite{Berlekamp64}.

\begin{lemma} \label{lem:13}
    For any protocol consisting of $A$ bits from Alice and $B$ bits from Bob, and for any three possible inputs $x_1, x_2, x_3$, there exists two of the three inputs $y_1, y_2 \in \{ x_1, x_2, x_3 \}$ and a transcript $T \in \{ 0, 1 \}^{A+B}$ such that the adversary can corrupt at most $\frac13 A +  1$ bits so that the protocol transcript is $T$ in both the case Alice has $y_1$ or $y_2$.
\end{lemma}

\begin{proof}
    We define $T$ as follows. 
    \begin{itemize}
        \item None of Bob's bits will be corrupted, that is, Alice will receive every bit that Bob sends correctly.
        \item As for Alice's bits, the adversary will begin by corrupting Alice's $t$'th bit to $\maj(a_t(x_1), a_t(x_2), a_t(x_3))$, where $a_t(x_i)$ denotes Alice's $t$'th bit if she has input $x_i$. For any $t \in [A]$, we denote by $\delta_t(x_i)$ the number of bits the adversary has corrupted up until the $t$'th bit if Alice has input $x_i$. The adversary continues this attack until the second largest of $\delta_t(x_i)$ reaches $\left\lceil\frac13 A\right\rceil$ in round $t_0$, at which point she switches to the following strategy: letting $y_1, y_2, y_3$ denote the $x_i$ with the smallest, second smallest, and largest value of $\delta_{t_0}(x_i)$ (so that $\delta_{t_0}(y_2) = \lceil \frac13 A \rceil$), the adversary corrupts the remainder of Alice's bits to be $a_t(y_2)$, for $t_0 < t \le A$. 
    \end{itemize}

    If point $t_0$ never happens, then the smallest and second smallest $\delta_A(x_i)$ are less than $\left\lceil\frac13 A\right\rceil$, and we are done. Otherwise, we claim that both $\delta_A(y_1)$ and $\delta_A(y_2)$ are at most $\left\lceil \frac13 A \right\rceil\leq \frac13 A + 1$ by at the end of the protocol. Clearly, $\delta_A(y_2) = \left\lceil \frac13A \right\rceil$. As for $\delta_A(y_1)$, we have that 
        \begin{align*}
            \delta_A(y_1) &\le \delta_{t_0}(y_1) + (A - t_0) \\
            &\le \delta_{t_0}(y_1) + A - \left( \delta_{t_0}(y_1) + \delta_{t_0}(y_2) + \delta_{t_0}(y_3) \right) \\
            &= A - \delta_{t_0}(y_2) - \delta_{t_0}(y_3) \\
            &\le A- 2\cdot\left\lceil\frac13 A\right\rceil \\
            &\le \frac13 A,
        \end{align*}
        where we use that $t_0 \ge \delta_{t_0}(y_1) + \delta_{t_0}(y_2) + \delta_{t_0}(y_3)$ (which holds because before point $t_0$ we are always causing corruption to at most one of the three transcripts), and that at point $t_0$ it holds that $\delta_{t_0}(y_3) \ge \delta_{t_0}(y_2) = \left\lceil\frac13 A\right\rceil$. 
\end{proof}

The attack is stated below.

\attack{1}{
    Let $x_1, x_2, x_3 \in \{ 0, 1 \}^k$ be three of Alice's possible inputs. By Lemma~\ref{lem:13}, there exist $y_1, y_2 \in \{ x_1, x_2, x_3 \}$ and transcript $T \in \{ 0, 1 \}^n$ such that the adversary can corrupt at most $\frac13 (A_1 + A_2) +1$ bits to obtain transcript $T$ in both the case Alice has $y_1$ and if she has $y_2$. The adversary simply corrupts the protocol so that the resulting transcript is $T$. 
    % for the entirety of the protocol, the adversary finds two of Alice's possible inputs $x_1,x_2$ and a transcript $T_B$ as in Lemma~\ref{lem:13} (note that Alice has at least three possible inputs). Then, she corrupts the protocol so that Bob receives $T_B$ using at most $\frac13(A_1+A_2)$ bits of corruption if Alice had $x_1,x_2$.
}

\begin{lemma} \label{lem:attack1}
    Attack~\ref{attack:1} succeeds with corrupting $ \frac13 A_1 + \frac13 A_2 + 1$ bits. 
\end{lemma}

\begin{proof}
    This follows immediately from Lemma~\ref{lem:13}: regardless of whether Alice has $y_1$ or $y_2$, the adversary is able to have Bob receive the same transcript, using $\frac13 \left( A_1 + A_2\right)+1$ bits of corruption.
\end{proof}
\subsection{Attack 2} \label{sec:attack2}

In our second attack, the adversary behaves differently in the two sections of the protocol. In the first section, the adversary essentially causes Bob's feedback to look random, so that Alice can do no better than to send a distance $\frac12$ error-correcting code. This allows the adversary to corrupt $\frac14$ of Alice's bits during this first section so that Bob cannot distinguish between three inputs. Then, in the second section, we use Lemma~\ref{lem:13} from the previous section to show that the adversary has a strategy corrupting only $\frac13 A_2$ bits to confuse Bob between two of the remaining three inputs.

To argue that the adversary can perform her attack in the first section, we need the following lemma.

\begin{lemma} \label{lem:14-12}
    For any $0<\epsilon<0.1$, suppose Alice has $K$ possible inputs where $K > (4/\epsilon)^{1/3}$. 
    % $K \ge e^{3\log (1/\epsilon) / \epsilon^3}$. 
    Then for any protocol consisting of $A$ bits from Alice and $B$ bits from Bob where $A + B \ge \frac{3 \log (1/\epsilon)}{\epsilon^3}$, there exist three inputs $x_1, x_2, x_3$ and a transcript $T$ such that regardless of which of $x_1, x_2, x_3$ Alice has as input, the adversary can corrupt at most $\left( \frac14 + \frac{3\epsilon}2 \right) \cdot A + \left( \frac12 + \epsilon \right) \cdot B + 1$ bits so that the protocol transcript is $T$.
\end{lemma}

\begin{proof}
    % We denote the $K$ possible inputs that Alice has by $x_1, \dots, x_K$. 
    Let $\Lambda$ denote the set of Alice's possible $K$ inputs.
    Suppose that the rounds of the $\iECC$ are structured so that when Alice is sending her $t$'th bit, she has seen $\gamma(t)$ bits from Bob so far. For $b \in \{ 0, 1 \}^{B}$ ($b$ can be thought of as what Alice receives from Bob throughout the entire protocol), we define the strings $a(y; b), a(y_1, y_2, y_3; b) \in \{ 0, 1 \}^{A}$ as follows:
    \begin{itemize}
        \item 
            We define $a(y; b)[t]$ to be the bit that Alice sends for her $t$'th bit if she has input $y$ and has seen $b[1:\gamma(t)]$ from Bob so far. We remark that we've abused notation here: Alice's message $a(y; b)[t]$ depends only on $b[1:\gamma(t)]$ and not on the rest of $b$, but we include all of $b$ for ease of notation.
        \item 
            For each (unordered) triple $(y_1, y_2, y_3)$ of different inputs, we define $a(y_1, y_2, y_3; b)[t]$ as follows. 
            \begin{itemize} 
                \item If $\max_{i \in \{ 1, 2, 3 \}} \Delta(a(y_1, y_2, y_3; b)[1:t-1], a(y_i)[1:t-1]) \le \left( \frac14 + \frac\epsilon2 \right) \cdot A - 1$, we set $a(y_1, y_2, y_3; b)[t] = \maj(a(y_1; b)[t], a(y_2; b)[t], a(y_3; b)[t])$. 
                \item Otherwise, let $i \in \{ 1,2,3 \}$ be such that $\Delta(a(y_1, y_2, y_3; b)[1:t-1], a(y_i)[1:t-1]) \ge \left( \frac14 + \frac\epsilon2 \right) \cdot A +1$ (if there's more than one value of $i$, take any), and set $a(y_1, y_2, y_3; b)[t] = a(y_i; b)[t]$.
            \end{itemize}
        % \item 
        %     Also, denote by $b(x_i, x_j, x_k; b) \in \{ 0, 1 \}^{B}$ to be Bob's responses if he sees $a(x_i, x_j, x_k; b)$ from Alice throughout the protocol. 
    \end{itemize}

    \begin{claim} \label{claim:14}
        For any $b \in \{ 0, 1 \}^{B}$, and for any three inputs $y_1, y_2, y_3 \in \Lambda$, if $\Diam(a(y_1; b), a(y_2; b), a(y_3; b)) \le \left( \frac12 + \epsilon \right) \cdot A$, then $\max_{i \in [3]} \Delta(a(y_1, y_2, y_3; b), a(y_i; b)) \le \left( \frac14 + \frac12\epsilon \right) \cdot A + 1$.
    \end{claim}

    \begin{proof}
        Suppose $a(y_1; b), a(y_2; b), a(y_3; b)$ are such that $\Diam(a(y_1; b), a(y_2; b), a(y_3; b)) \le \left( \frac12 + \epsilon \right) \cdot A$. Let $T_{123}$ be the indices $t \in [A]$ on which $a(y_1; b)[t] = a(y_2; b)[t] = a(y_3; b)[t]$. Also let $T_{12}$ be the indices $t$ for which $a(y_1; b)[t] = a(y_2; b)[t] \not= a(y_3; b)[t]$, and define $T_{23}$ and $T_{31}$ analogously. Note that $T_{123}, T_{12}, T_{23}, T_{31}$ are all disjoint and together cover all indices $[A]$. 

        If $|T_{12}|, |T_{23}|, T_{31}| < \left( \frac14 + \frac\epsilon2 \right) \cdot A+1$, then notice that $a(y_1, y_2, y_3; b)[t] = \maj(a(y_1; b)[t], a(y_2; b)[t], a(y_3; b)[t])$ for all $t \in [A]$, so $\Delta(a(y_1, y_2, y_3; b), a(y_i; b)) = |T_{-i}| \le \left( \frac14 + \epsilon \right) \cdot A$, where the $-i$ in $T_{-i}$ is the other two of $\{ 1, 2, 3 \}$ not equal to $i$. 

        Otherwise, at most one of $|T_{12}|, |T_{23}|, |T_{31}|$ is greater than $\left( \frac14 + \frac\epsilon2 \right) \cdot A$, since if e.g. $|T_{12}|, |T_{23}| > \left( \frac14 + \frac\epsilon2 \right) \cdot A$, then $\Delta(a(y_1; b), a(y_3; b)) = |T_{12}| + |T_{23}| > \left( \frac12 + \epsilon \right) \cdot A$, which is a contradiction. Now, suppose without loss of generality that $|T_{12}| > \left( \frac14 + \frac\epsilon2 \right) \cdot A$. In the above algorithm for setting $a(y_1, y_2, y_3; b)[t]$, note that we set $a(y_1, y_2, y_3; b)[t] = \maj(a(y_1; b)[t], a(y_2; b)[t], a(y_3; b)[t])$ up until the $\left\lceil \left( \frac14 + \epsilon \right) \cdot A \right\rceil$'th smallest index $t_0 \in T_{12}$. After that, we set $a(y_1, y_2, y_3; b)[t] = a(y_3; b)[t]$. Thus, $\Delta(a(y_1, y_2, y_3; b), a(y_3; b)) = \left\lceil \left( \frac14 + \frac\epsilon2 \right) \cdot A \right\rceil$. Note that there are $|T_{12}| - \left\lceil \left( \frac14 + \frac\epsilon2 \right) \cdot A \right\rceil$ more indices of $T_{12}$ for which $a(y_1, y_2, y_3; b)[t]$ will be set to $a(y_3; b)[t]$ instead of $a(x_1; b)[t] = a(x_2; b)[t]$. Since $\left( |T_{12}| - \left\lceil \left( \frac14 + \frac\epsilon2 \right) \cdot A \right\rceil \right) + |T_{-i}| = \left( |T_{12}| + |T_{-i}| \right) - \left\lceil \left( \frac14 + \frac\epsilon2 \right) \cdot A \right\rceil \le \left( \frac12 + \epsilon \right) \cdot A - \left\lceil \left( \frac14 + \frac\epsilon2 \right) \cdot A \right\rceil \le \left( \frac14 + \frac\epsilon2 \right) \cdot A$ for $i = 1,2$, it follows that for any index $t \in T_{23}$, $a(y_1, y_2, y_3; b)[t]$ will be set to $a(y_2; b)[t] = a(y_3; b)[t]$ and similarly for any $t \in T_{31}$, $a(y_1, y_2, y_3; b)[t]$ will be set to $a(y_3; b)[t] = a(y_1; b)[t]$. Then $\Delta(a(y_1, y_2, y_3; b), a(y_i; b)) \le \left( \frac14 + \frac\epsilon2 \right) \cdot A + 1$ for $i = 1,2$ as well.
    \end{proof}

    We now split the proof of Lemma~\ref{lem:14-12} into two cases depending on the size of $B$ relative to $A$. If $B \le \epsilon \cdot (A+B)$, then consider the following attack: the adversary chooses any string $b \in \{ 0, 1 \}^{B}$ and will corrupt Bob's communication so that Alice receives $b$. This takes at most $B \le \epsilon \cdot (A+B)$ corruptions. By Corollary~\ref{cor:shearer}, there exists three inputs $x_1, x_2, x_3 \in \Lambda$ for which $\Diam(a(x_1; b), a(x_2; b), a(x_3; b)) \le \left( \frac12 + \epsilon \right) \cdot A$, so by Claim~\ref{claim:14}, it takes at most $\left( \frac14 + \frac\epsilon2 \right) \cdot A + 1$ corruptions to corrupt Alice's messages to $a(x_1, x_2, x_3; b)$ from any of $a(x_1; b), a(x_2; b), a(x_3; b)$. This is a total of at most $\left( \frac14 + \frac{3\epsilon}{2} \right) \cdot A + \epsilon B + 1$ corruptions. 

    The second case is if $B > \epsilon \cdot (A+B)$. In this case, the adversary picks $(x_1, x_2, x_3)$ as follows. By Corollary~\ref{cor:shearer}, we have that
    \[
        \Pr_{\substack{(y_1, y_2, y_3) \in \binom{\Lambda}{3} \\ b \in \{ 0, 1 \}^{B}}} \left[ \Diam(a(y_1; b), a(y_2; b), a(y_3; b)) \le \left( \frac12 + \epsilon \right) \cdot A\right] \ge \frac{\epsilon}{4},
    \]
    so there exists $(x_1, x_2, x_3)$ such that 
    \[
        \Pr_{b \in \{ 0, 1 \}^{B}} \left[ \Diam(a(x_1; b), a(x_2; b), a(x_3; b)) \le \left( \frac12 + \epsilon \right) \cdot A\right] \ge \frac{\epsilon}{4}.
    \]
    In particular, there exists a set $\Pi \subseteq \{ 0, 1 \}^B$ of size $|\Pi| \ge \epsilon \cdot 2^{B-2}$ such that for any $b \in \Pi$ it holds that $\Diam(a(x_1; b), a(x_2; b), a(x_3; b)) \le \left( \frac12 + \epsilon \right) \cdot A$.

    Now, consider the following process: For each of Bob's bits, the adversary flips it independently with probability $\frac12$. This defines (at any point, the prefix of) a string $b \in \{ 0, 1 \}^B$. Note that Alice receives any string $b \in \{ 0, 1 \}$ with probability $2^{-B}$ since every of Bob's bit is flipped with $\frac12$ probability. Meanwhile, the adversary corrupts Alice's bits to $a(x_1, x_2, z_3; b)$ (recall that $a(x_1, x_2, z_3; b)[t]$ depends only on $b[1:\gamma(t)]$ and not on all of $b$ so this attack is well defined).

    For notation, let $\beta$ be the random variable denoting what Bob sends throughout this process. We have that
    \begin{align*}
        \Pr_{b} &
        \left[ 
            \left( \max_{i \in [3]} \Delta(a(x_i; b), a(x_1, x_2, x_3; b)) \le \left( \frac14 + \frac\epsilon2 \right) \cdot A \right)
            \wedge \left( \Delta(b, \beta) \le \left( \frac12 + \epsilon \right) \cdot B \right)
        \right] \\
        &\ge \Pr_{b} \left[ b \in \Pi \right] - \Pr_{b} \left[ \Delta(b, \beta) \le \left( \frac12 + \epsilon \right) \cdot B \right] \\
        &\ge \frac{\epsilon}{4} - e^{-2\epsilon^2 B/3} \\
        &> \frac{\epsilon}{4} - e^{-2\epsilon^3 (A+B)/3},
    \end{align*}
    where the first inequality follows from the fact that by Claim~\ref{claim:14}, $\max_{i \in [3]} \Delta(a(x_i; b), a(x_1, x_2, x_3; b)) \le \left( \frac14 + \frac\epsilon2 \right) \cdot A + 1$ holds whenever $b \in \Pi$, the second inequality follows from the Chernoff bound, and the last inequality follows from our assumption that $B > \epsilon \cdot (A+B)$. Since $\frac{\epsilon}{4} > \epsilon^2 = e^{-2\log(1/\epsilon)} \ge e^{-2\epsilon^3(A+B)/3}$, this expression is positive, and so there exists a choice of $b \in \{ 0, 1 \}^B$ for which $\max_{i \in [3]} \Delta(a(x_i; b), a(x_1, x_2, x_3; b)) \le \left( \frac14 + \frac\epsilon2 \right) \cdot A + 1$ and $\Delta(b, \beta) \le \left( \frac12 + \epsilon \right) \cdot B$. Then if the adversary corrupts Bob's messages to $b$ and Alice's messages to $a(x_1, x_2, x_3; b)$, the total corruption necessary regardless of which of $x_1, x_2, x_3$ Alice has is at most $\left( \frac14 + \frac\epsilon2 \right) \cdot A + \left( \frac12 + \epsilon \right) \cdot B$ + 1.

\end{proof}

We now state our second attack.

\attack{2}{

    Let inputs $x_1, x_2, x_3 \in \{ 0, 1 \}^k$ and transcript $T_1 \in \{ 0, 1 \}^{21n/47}$ be such that they satisfy Lemma~\ref{lem:14-12} for the first section of the protocol. Then, for the first section of the protocol, the adversary corrupts the transcript to look like $T_1$, using at most $\left( \frac14 + \frac{3\epsilon}{2} \right) A_1 + \left( \frac12 + \epsilon \right) B_1 + 1$ bits of corruption in the cases where Alice had $x_1,x_2,x_3$.

    \medskip
    
    For the second section of the protocol, the adversary corrupts the communication to the transcript $T_2 \in \{ 0, 1 \}^{26n/47}$ as found in Lemma~\ref{lem:13} such that there exist two of $x_1, x_2, x_3$, denoted $y_1, y_2$, for which the adversary can corrupt at most $\frac13 A_2+1$ of the communication so that the transcript of received bits is $T_2$ if Alice has $y_1$ or $y_2$.
}

\begin{lemma} \label{lem:attack2}
    Suppose that $k \ge \frac{141\log(1/\epsilon)}{21\epsilon^3} \cdot $ Attack~\ref{attack:2} succeeds with corrupting $\left( \frac14 + \frac{3\epsilon}2 \right) \cdot A_1 + \left(\frac12 + \epsilon\right) \cdot B_1 + \frac13 A_2+2$ bits. 
\end{lemma}

\begin{proof}
    Regardless of whether Alice has $y_1$ or $y_2$, the transcript from Bob's perspective when the adversary employs this attack looks like $T_1$ followed by $T_2$ (restricted to Bob's viewpoint). Since $A_1 + B_1 \ge \max \{ \frac{21}{26} A_2, k - A_2 \}$ (where $A_1 + A_2 \ge k$ holds since Alice needs to send $k$ bits to communicate $x$, even noiselessly), it follows that $A_1 + B_1 \ge \frac{21}{47} k \ge 3\log(1/\epsilon)/\epsilon^3$, so the condition of Lemma~\ref{lem:14-12} is satisfied. Then, by Lemma~\ref{lem:14-12}, the number of corruptions used in the first section of the protocol when Alice has $y_1$ or $y_2$ is at most $\left( \frac14 + \frac{3\epsilon}{2} \right) \cdot A_1 + \left( \frac12 + \epsilon \right) \cdot B_1$, and by Lemma~\ref{lem:13}, the number of corrupted bits in the second section whether Alice has $y_1$ or $y_2$ is at most $\frac13 A_2+1$. 
\end{proof}
\subsection{Attack 3} \label{sec:attack3}

In our third attack, we employ the following strategy. At a high level, we choose two inputs $x_1$ and $x_2$. In the first section of the protocol, Bob's view is as if Alice had $x_1$, while Bob's bits are corrupted so that Alice thinks that he has been receiving and responding correctly. In the second section of the protocol, Bob's bits are flipped randomly, and Alice's communication is corrupted to look like she has $x_2$. 

% \todo{I think to make this proof/statements more rigorous we have to define $\epsilon$s and input space size $N$ more carefully; I'm ok with ignoring this or adding in the rigor. also comments aren't working fsr}

The first lemma we will need is to show that for the first section of the protocol, there are many inputs for which the uncorrupted transcripts have pairwise small Hamming distance.

\begin{lemma} \label{lem:big-set}
    Let $\epsilon > 0$ and suppose Alice has $K$ possible inputs. Then for any protocol consisting of $A$ bits from Alice and $B$ bits from Bob, there exists a set $\Gamma$ of size $K'_\epsilon(K) = K^{\epsilon} - \frac1\epsilon$ inputs such that for any two $x_1, x_2 \in \Gamma$, the relative distance of the (uncorrupted) transcripts in the case where Alice has $x_1$ or $x_2$ is $\le \left( \frac12 + \epsilon \right) \cdot (A+B)$.
\end{lemma}

\begin{proof} 
    Consider a graph where the $K$ possible inputs are the vertices, and draw an edge from $x$ to $y$ if $\Delta(T(x), T(y)) \le \left( \frac12 + \epsilon \right) \cdot (A+B)$, where $T(z)$ denotes the (uncorrupted) transcript corresponding to when Alice has input $z$. Then by Lemma~\ref{lem:clique-dist-12}, there does not exist an independent set of size $\frac{1}{\epsilon}$. Then by Theorem~\ref{thm:ramsey}, if there didn't exist an $K'_\epsilon(K)$-clique, then $K < R(K'_\epsilon(K), \frac{1}{\epsilon}) \le \binom{K'_\epsilon(K) + 1/\epsilon}{1/\epsilon} \le \left( K'_\epsilon(K) + \frac1\epsilon \right)^{1/\epsilon} = K$, contradiction.
\end{proof}

\begin{lemma} \label{lem:12-0}
    Let $\epsilon > 0$, and suppose Alice has $K' > \sqrt{2/\epsilon}$ possible inputs. For any protocol consisting of $A$ bits from Alice and $B$ bits from Bob such that $A+B > \frac{3\log(1/\epsilon)}{\epsilon^3}$, there exist two inputs $x_1, x_2$ such that for any advice $\alpha$ that Bob receives at the beginning of the protocol (after both Alice and Bob have fixed their strategies), there exist two transcripts $T_1, T_2$ such that the Bob's view of the two transcripts is the same, and that in the case of Alice having $x_1$, the adversary needs only corrupt $\left( \frac12 + 2\epsilon \right) A + \left( \frac12 + \epsilon \right) B$ bits to get transcript $T_1$, and in the case of Alice having $x_2$, the adversary needs only corrupt $\left( \frac12 + \epsilon \right) B$ bits so that the transcript is $T_2$. 
\end{lemma}

\begin{proof} 
    Let $\Lambda$ be all of Alice's possible inputs. Suppose that when Alice is sending the $t$'th bit, she has seen $\gamma(t)$ bits from Bob so far. For input $y$ and $b \in \{ 0, 1 \}^B$ ($b$ can be thought of what Alice receives from Bob throughout the protocol), we define the string $a(y; b)$ as follows: $a(y; b)[t]$ is what Alice would send for her $t$'th bit if she has $y$ as input and has seen $b[1:\gamma(t)]$ from Bob so far. We remark Alice's $t$'th bit depends only on $b[1:\gamma(t)]$ and not on the rest of $b$, but we've included all of $b$ for ease of notation. 

    If Bob speaks for $B \le \epsilon ( A + B)$ bits, then consider the following attack: the adversary chooses some string $b \in \{ 0, 1 \}^B$ and will corrupt Bob's bits to look like $b$, requiring at most $B \le \epsilon (A+B)$ corruption. Let $a(z; b)$ be Alice's bits if she has input $z$ and receives $b$ throughout the protocol. Then, by Corollary~\ref{cor:turan}, there exist two inputs $x_1, x_2$ such that $\Delta(a(x_1; b), a(x_2; b)) < \left( \frac12 + \epsilon \right) \cdot A$. The adversary can corrupt the Alice's communication in both cases to look like $a(x_2; b)$, requiring a total of $\le \left( \frac12 + \epsilon \right) \cdot A + \epsilon (A+B) = \left( \frac12 + 2\epsilon \right) \cdot A + \epsilon B$ corruption in the case that Alice has $x_1$, and $\le \epsilon (A+B)$ corruption in the case that Alice has $x_2$.  

    Otherwise, suppose for the remainder of this proof that Bob speaks for $B > \epsilon (A+B)$ bits. By Corollary~\ref{cor:turan}, we have that
    \[
        \Pr_{\substack{ (y_1, y_2) \in \binom{\Lambda}{2} \\ b \in \{ 0, 1 \}^B }} \left[ \Delta(a(y_1; b), a(y_2; b)) \le \left( \frac12 + \epsilon \right) \cdot A \right] \ge \frac\epsilon2,
    \]
    so there exists $(x_1, x_2) \in \binom{\Lambda}{2}$ for which 
    \[
        \Pr_{ b \in \{ 0, 1 \}^B } \left[ \Delta(a(x_1; b), a(x_2; b)) \le \left( \frac12 + \epsilon \right) \cdot A \right] \ge \frac\epsilon2.
    \]
    In particular, there exists a set $\Pi \subseteq \{ 0, 1 \}^B$ of size $|\Pi| \ge \epsilon \cdot 2^{B-1}$ such that for any $b \in \Pi$, it holds that $\Delta(a(x_1; b), a(x_2; b)) \le \left( \frac12 + \epsilon \right) \cdot A$. Note that this choice of $(x_1, x_2)$ is independent of any advice $\alpha$ that Bob may have received.
    
    Consider the following process: for each bit that Bob sends, we flip it with probability $\frac12$. This defines (at any point, the prefix of) a string $b \in \{ 0, 1 \}^B$. Note that Alice receives any string $b \in \{ 0, 1 \}^B$ with probability $2^{-B}$ since each of Bob's bits are flipped with probability $\frac12$. Meanwhile, the adversary corrupts Alice's string to $a(x_2; b)$ (note that $a(x_2; b)[t]$ depends only on $b[1:\gamma(t)]$).

    For notation, let $\beta$ denote what Bob sends throughout this process. We have that 
    \begin{align*}
        \Pr_b & \left[ 
        \left( \Delta(a(x_1; b), a(x_2; b)) \le \left( \frac12 + \epsilon \right) \cdot A \right) 
        \wedge \left( \Delta(b, \beta) \le \left( \frac12 + \epsilon \right) \cdot B \right)
        \right] \\
        & \ge \Pr_b [b \in \Pi] - \Pr_b \left[ \Delta(b, \beta) \le \left( \frac12 + \epsilon \right) \cdot B \right] \\
        & \ge \frac\epsilon2 - e^{-2\epsilon^2 B/3} \\
        &\ge \frac\epsilon2 - e^{-2\epsilon^3 (A+B)/3}.
    \end{align*}
    We have that $\frac\epsilon2 > \epsilon^2 = e^{-2\log(1/\epsilon)} > e^{-2\epsilon^3(A+B)/3}$, so this expression is positive, and so there exists a choice of $b \in \{ 0, 1 \}^B$ for which $\Delta(a(x_1; b), a(x_2; b)) \le \left( \frac12 + \epsilon \right) \cdot A$ and $\Delta(b, \beta) \le \left( \frac12 + \epsilon \right) \cdot B$. In other words, this attack in the case of Alice having $x_1$ uses at most $\left( \frac12 + \epsilon \right) \cdot A + \left( \frac12 + \epsilon \right) \cdot B$ corruption, and in the case of Alice having $x_2$, it uses $\left( \frac12 + \epsilon \right) \cdot B$ corruption.
    
\end{proof}

\attack{3}{
    Denote by $T_1(y)$ the uncorrupted transcript corresponding to Alice having input $y$ in the first section of the protocol. By Lemma~\ref{lem:big-set}, there exists a set $M$ of $2^{\epsilon k} - \frac1\epsilon$ inputs such that for every $y_1, y_2 \in M$, it holds that $\Delta(T_1(y_1), T_1(y_2)) \le \left( \frac12 + \epsilon \right) \cdot n$. Next, consider the second section of the protocol, conditioned on Alice having seen $T_1(x)$ (restricted to her view) in the first section of the protocol. By Lemma~\ref{lem:12-0} there exist $x_1, x_2 \in M$ such that no matter what advice $\alpha$ Bob receives at the beginning of this second section, there exist transcripts $T_{2,1}(\alpha)$ and $T_{2,2}(\alpha)$ such that Bob's view of the two transcripts are the same, and these $T_{2,1}(\alpha), T_{2,2}(\alpha)$ satisfy the properties listed in Lemma~\ref{lem:12-0}.

    \medskip

    In the first section of the protocol, the adversary corrupts the communication so that Bob always receives $T_1(x_1)$ (restricted to the bits that Bob sees), and so that Alice receives $T(x)$ (restricted to the bits that she sees), where $x$ denotes Alice's input. 

    \medskip 

    In the second section of the protocol, the adversary corrupts the communication so that the transcript is $T_{2,1}(\alpha = T_1(x_1))$ in the case of Alice having $x_1$, and $T_{2,2}(\alpha = T_1(x_1))$ otherwise. 

    % \medskip
    
    % Find the set $M$ as in Lemma~\ref{lem:big-set}. Then, find $x_1,x_2\in M$ and $T$ satisfying Lemma~\ref{lem:12-0} in the second half of the protocol. This applies because in the first half of the protocol Bob may learn $x_1,x_2$, but Alice's transcript so far is only dependent on her input.
    
    % \medskip
    
    % Let $\tau(x)$ be the transcript if Alice has input $x$. For the first section of the protocol, the adversary corrupts both cases to the following transcript: Alice receives bits corresponding to $\tau(x)$ and Bob receives bits corresponding to $\tau(x_2)$. This uses $\frac12 A_1 + \frac12 B_1$ bits of corruption in the case where Alice has $x_1$ and no corruption otherwise.

    % \medskip
    
    % For the second section, the adversary corrupts messages as in the lemma so that the overall transcript for the remaining part is $T$, using $\frac12 B_2$ bits if Alice has $x_1$ and $\frac12 A_2 + \frac12 B_2$ bits if Alice has $x_2$.
}

\begin{lemma} \label{lem:attack3}
    Suppose that $k > \frac{141\log(1/\epsilon)}{26\epsilon^3}$. Attack~\ref{attack:3} succeeds with corrupting 
    \[
        \max \left\{ \left( \frac12 + 2\epsilon \right) \cdot A_2 + \left( \frac12 + \epsilon \right) \cdot B_2~,~~ \left( \frac12 + \epsilon \right) \cdot A_1 + \left( \frac12 + \epsilon \right) \cdot B_1 + \left( \frac12 + \epsilon \right) \cdot B_2 \right\}
    \]
    bits. 
\end{lemma}

\begin{proof}
    Regardless of whether Alice has $x_1$ or $x_2$, Bob receives the same transcript (restricted to his view). Since $A_2 + B_2 \ge \max \{ \frac{26}{21} A_1, k - A_1 \}$ (where $A_1 + A_2 \ge k$ holds since Alice needs to send $k$ bits to communicate $x$, even noiselessly), it follows that $A_2 + B_2 \ge \frac{26}{47} k > \frac{3\log(1/\epsilon)}{\epsilon^3}$, so the condition of Lemma~\ref{lem:12-0} is satisfied. 
    If Alice has $x_1$, the amount of corruption in the first section is $0$, while in the second section the adversary corrupted at most $\left( \frac12 + 2\epsilon \right) A_2 + \left( \frac12 + \epsilon \right) \cdot B_2$ bits. If Alice has $x_2$, the amount of corruption in the first section is $\Delta(T_1(x_1), T_1(x_2)) \le \left( \frac12 + \epsilon \right) \cdot (A_1 + B_1)$, and in the second section the adversary corrupts at most $\left( \frac12 + \epsilon \right) \cdot B_2$ bits.
    % Regardless of whether Alice has $x_1$ or $x_2$, the adversary is able to have Bob receive the same transcript. If Alice has $x_1$, this costs $\frac12 A_1 + \frac12 B_1 + \frac12 B_2$ bits of corruption, and if Alice has $x_2$ this costs $\frac12 A_2 + \frac12 B_2$ bits of corruption. Overall, Eve corrupts $\max \left\{ \frac12 A_1 + \frac12 B_1 + \frac12 B_2, \frac12 A_2 + \frac12 B_2 \right\}$ bits. 
\end{proof}

\subsection{Proof of Theorem~\ref{thm:13/47}} \label{sec:13/47-proof}

In this section, we prove our main theorem, restated below.

\main*

We begin with the following lemma.

\begin{lemma} \label{lem:ineq}
    For any nonnegative $a_1,b_1,a_2,b_2 \in \bbR$ where $a_1+b_1=\frac{21}{47}$ and $a_1+b_1+a_2+b_2=1$, define 
    \begin{align*}
        \delta_1 &= \frac13 a_1 + \frac13 a_2, \\
        \delta_2 &= \frac14 a_1 + \frac12 b_1 + \frac13 a_2, \\
        \delta_3 &= \max \left\{ \frac12 a_2 + \frac12 b_2, \frac12 a_1 + \frac12 b_1 + \frac12 b_2 \right\}.
    \end{align*}
    It holds that 
    \[
        \min \{ \delta_1, \delta_2, \delta_3 \} \le \frac{13}{47}.
    \]
\end{lemma}

\begin{proof}
    Using that $b_1=\frac{21}{47}-a_1$ and $b_2=\frac{26}{47}-a_2$, we can substitute:
    \begin{align*}
        \delta_1 &= \frac13 a_1 + \frac13 a_2, \\
        \delta_2 &= \frac{21}{94} - \frac14 a_1 + \frac13 a_2, \\
        \delta_3 &= \max \left\{ \frac{13}{37}, \frac12 - \frac12 a_2 \right\}.
    \end{align*}
    Then, 
    \[
        \min \{ \delta_1, \delta_2, \delta_3 \} \le \frac{13}{47} \Longleftrightarrow \min \{ \delta_1, \delta_2, \delta'_3 \} \le \frac{13}{47},
    \]
    where $\delta'_3 = \frac12 - \frac12 a_2$. But note that 
    \[
        \frac{9}{35} \delta_1 + \frac{12}{35} \delta_2 + \frac{2}{5} \delta'_3 
        = \frac{13}{47}
    \]
    where the weights $\frac{9}{35}, \frac{12}{35}, \frac25$ sum to $1$, so at least one of $\delta_1, \delta_2, \delta'_3$ must be at most $\frac{13}{47}$. 
    % \begin{align*}
    %     \frac{13}{47}\geq 
    %     \min \left\{
    %     \begin{aligned}
    %         \frac{21}{94} - \frac14 a_1 + \frac13 a_2, \\
    %         \frac13 a_1 + \frac13 a_2, \\
    %         \frac12 - \frac12 a_2
    %     \end{aligned}
    %     \right\}
    % \end{align*}

    % Adding the equations in the weights $\frac{12}{35},\frac{9}{35},\frac25$ and noting that the minimum of the right-hand side is at most any positive linear combination with weights summing to $1$, we get the desired inequality.

\end{proof}

\begin{proof}[Proof of Theorem~\ref{thm:13/47}]
    Recall that the length of the $\iECC$ is $n:=A_1+B_1+A_2+B_2 \ge A_1 + A_2 \geq k>\epsilon^{-3}$ (since Alice needs to send at least $k$ bits to communicate $x$, even in the noiseless setting). Our goal is to show that regardless of the values of $A_1, B_1, A_2, B_2$, at least one of Attacks~\ref{attack:1},~\ref{attack:2}, and~\ref{attack:3} will require at most $\left( \frac{13}{47} + 2\epsilon \right) \cdot n$ corruptions.

    By Lemma~\ref{lem:attack1}, Attack~\ref{attack:1} succeeds using 
    \[
        \frac13 A_1 + \frac13 A_2 + 1 \le \left( \delta_1 + 2\epsilon \right) n
    \]
    bits of corruption, where $\delta_1 := (\frac13 A_1 + \frac13 A_2)/n$. 

    By Lemma~\ref{lem:attack2}, Attack~\ref{attack:2} succeeds using 
    \[
        \left( \frac14 + \frac{3\epsilon}{2} \right) \cdot A_1 + \left( \frac12 + \epsilon \right) \cdot B_1 + \frac13 A_2 + 2 \le \frac14 A_1 + \frac12 B_1 + \frac13A_2 + 2\epsilon n = (\delta_2 + 2\epsilon) n
    \]
    bits of corruption, where we define $\delta_2 = (\frac14 A_1 + \frac12 B_1 + \frac13A_2)/n$.

    By Lemma~\ref{lem:attack3}, Attack~\ref{attack:3} succeeds using 
    \begin{align*}
        \max & \left\{ \left( \frac12 + 2\epsilon \right) \cdot A_2 + \left( \frac12 + \epsilon \right) \cdot B_2~,~~ \left( \frac12 + \epsilon \right) \cdot A_1 + \left( \frac12 + \epsilon \right) \cdot B_1 + \left( \frac12 + \epsilon \right) \cdot B_2 \right\} \\
        &\le \max \left\{ \frac12 A_2 + \frac12 B_2 +2\epsilon n, \frac12 A_1 + \frac12 B_1 + \frac12 B_2 + 2\epsilon n  \right\} \\
        &= (\delta_3 + 2\epsilon) n
    \end{align*}
    bits of corruption, where we define $\delta_3 = \max \left\{ \frac12 A_2 + \frac12 B_2, \frac12 A_1 + \frac12 B_1 + \frac12 B_2  \right\} / n$. 

    By Lemma~\ref{lem:ineq}, we have that $\min \{ \delta_1, \delta_2, \delta_3 \} \le \frac{13}{47}$, so at least one of the three attacks succeeds with $\left( \frac{13}{47} + 2\epsilon \right) \cdot n$ corruption, regardless of the relative ratios of $A_1, B_1, A_2, B_2$.

\end{proof}

\bibliographystyle{alpha}
\bibliography{refs}

\newcommand{\etalchar}[1]{$^{#1}$}
\begin{thebibliography}{BEK{\etalchar{+}}22}

\bibitem[ADL06]{AhlswedeDL06}
Rudolf Ahlswede, Christian Deppe, and Vladimir Lebedev.
\newblock {Non-binary error correcting codes with noiseless feedback, localized
  errors, or both}.
\newblock In {\em 2006 IEEE International Symposium on Information Theory},
  pages 2486--2487, 2006.

\bibitem[BEK{\etalchar{+}}22]{BravermanEKSZ22}
Mark Braverman, Klim Efremenko, Gillat Kol, Raghuvansh Saxena, and Zhijun
  Zhang.
\newblock Round-vs-resilience tradeoffs for binary feedback channels.
\newblock {\em Electronic Colloquium on Computational Complexity}, TR22-179,
  December 2022.

\bibitem[Ber64]{Berlekamp64}
Elwyn~R. Berlekamp.
\newblock {Block coding with noiseless feedback}.
\newblock 1964.

\bibitem[Ber68]{Berlekamp68}
Elwyn~R. Berlekamp.
\newblock {Block coding for the binary symmetric channel with noiseless,
  delayless feedback}.
\newblock {\em Error-correcting Codes}, pages 61--88, 1968.

\bibitem[BY08a]{BurnashevY08a}
Marat Burnashev and Hirosuke Yamamoto.
\newblock On the zero-rate error exponent for a bsc with noisy feedback.
\newblock {\em Problems of Information Transmission}, 44, 09 2008.

\bibitem[BY08b]{BurnashevY08b}
Marat~V. Burnashev and Hirosuke Yamamoto.
\newblock On bsc, noisy feedback and three messages.
\newblock In {\em 2008 IEEE International Symposium on Information Theory},
  pages 886--889, 2008.

\bibitem[EKSZ22]{EfremenkoKSZ22}
Klim Efremenko, Gillat Kol, Raghuvansh Saxena, and Zhijun Zhang.
\newblock Binary codes with resilience beyond 1/4 via interaction.
\newblock {\em Proceedings - Annual IEEE Symposium on Foundations of Computer
  Science, FOCS}, 2022.

\bibitem[GGZ22]{GuptaGZ22}
Meghal Gupta, Venkatesan Guruswami, and Rachel~Yun Zhang.
\newblock Binary error-correcting codes with minimal noiseless feedback.
\newblock {\em To appear in STOC 2023}, 2022.

\bibitem[GKZ22]{GuptaKZ22}
Meghal Gupta, Yael~Tauman Kalai, and Rachel~Yun Zhang.
\newblock Interactive error correcting codes over binary erasure channels
  resilient to > $1/2$ adversarial corruption.
\newblock In {\em Proceedings of the 54th Annual ACM SIGACT Symposium on Theory
  of Computing}, pages 609--622, 2022.

\bibitem[GZ22]{GuptaZ22B}
Meghal Gupta and Rachel~Yun Zhang.
\newblock Positive rate binary interactive error correcting codes resilient to
  $1/2 $ adversarial erasures.
\newblock {\em arXiv preprint arXiv:2201.11929}, 2022.

\bibitem[Ham50]{Hamming50}
R.~W. Hamming.
\newblock {Error detecting and error correcting codes}.
\newblock {\em The Bell System Technical Journal}, 29(2):147--160, 1950.

\bibitem[HKV15]{HaeuplerKV15}
Bernhard Haeupler, Pritish Kamath, and Ameya Velingker.
\newblock {Communication with Partial Noiseless Feedback}.
\newblock In {\em APPROX-RANDOM}, 2015.

\bibitem[Ram87]{Ramsey87}
Frank~P Ramsey.
\newblock On a problem of formal logic.
\newblock {\em Classic Papers in Combinatorics}, pages 1--24, 1987.

\bibitem[Sha48]{Shannon48}
Claude~E. Shannon.
\newblock {A mathematical theory of communication}.
\newblock {\em The Bell System Technical Journal}, 27(3):379--423, 1948.

\bibitem[She83]{Shearer83}
James~B. Shearer.
\newblock A note on the independence number of triangle-free graphs.
\newblock {\em Discrete Math.}, 46(1):83–87, jan 1983.

\bibitem[SW92]{SpencerW92}
Joel Spencer and Peter Winkler.
\newblock {Three Thresholds for a Liar}.
\newblock {\em Combinatorics, Probability and Computing}, 1(1):81–93, 1992.

\bibitem[Tur41]{Turan41}
Paul Turán.
\newblock On an extremal problem in graph theory.
\newblock {\em Matematikai és Fizikai Lapok}, 48:436–452, 1941.

\bibitem[WQC17]{WangQC17}
Gang Wang, Yanyuan Qin, and Chengjuan Chang.
\newblock {Communication with partial noisy feedback}.
\newblock In {\em 2017 IEEE Symposium on Computers and Communications (ISCC)},
  pages 602--607, 2017.

\bibitem[Zig76]{Zigangirov76}
K.Sh. Zigangirov.
\newblock {Number of correctable errors for transmission over a binary
  symmetrical channel with feedback}.
\newblock {\em Problems Inform. Transmission}, 12:85--97, 1976.

\end{thebibliography}

\end{document}